\documentclass[12pt]{article}
\usepackage{amsmath,amssymb,amsthm}

\newtheorem{theorem}{Theorem}
\newtheorem{corollary}{Corollary}

\newcounter{example}
\newenvironment{example}[1][]{\refstepcounter{example}\par\medskip\textbf{Example~\theexample. #1}}{\medskip}

\def\tr{{\rm tr\,}}

\title{From Golden to Unimodular Cryptography}
\author{Sergiy Koshkin and Taylor Styers\\
Department of Mathematics and Statistics\\
 University of Houston-Downtown\\
 One Main Street\\
 Houston, TX 77002\\
 e-mail: koshkins@uhd.edu}
\date{}

\begin{document}

\maketitle
\begin{abstract}\
We introduce a natural generalization of the golden cryptography, which uses general unimodular matrices in place of the traditional $Q$ matrices, and prove that it preserves the original error correction properties of the encryption. Moreover, the additional parameters involved in generating the coding matrices make this unimodular cryptography resilient to the chosen plaintext attacks that worked against the golden cryptography. Finally, we show that even the golden cryptography is generally unable to correct double errors in the same row of the ciphertext matrix, and offer an additional check number which, if transmitted, allows for the correction.
 
\bigskip

\end{abstract}

\section{Introduction}

In a number of papers \cite{St78,St06,St07} and books \cite{StFib, StHar} Stakhov developed the so-called ``golden cryptography", a system of encryption based on utilizing matrices with entries being consecutive Fibonacci numbers with fast encryption/decryption and nice error-correction properties. It was applied to creating digital signatures \cite{BAh}, and somewhat similar techniques are used for image encryption and scrambling \cite{Mish}. However, the golden cryptography and its generalizations are vulnerable to the chosen plaintext attacks, which make it insecure \cite{ReySan,Tah}. As pointed out in \cite{Tah}, this is due to rigidity of the scheme, it has too few parameters to make the generation of coding matrices hard to backtrack. Their suggestion was to use an external hash function to add the extra parameters, multiple encryption was suggested in \cite{Sud}, and the use of Haar wavelets in \cite{Moh}. 

We propose a much more intrinsic generalization of the golden cryptography, which increases the number of free parameters while preserving its error correction properties. The idea is to use an arbitrary unimodular matrix in place of the so-called golden matrix $Q$. There are some mathematical subtleties with implementing this idea, which we work out. In particular, we prove that under mild assumptions there is an analog of the golden ratio for unimodular matrices, what we call the {\it unimodular ratio}, and the ratios of the coding matrix entries are close to it, a key property exploited for error correction. While most types of errors are then correctable as in the golden cryptography, it turns out that double errors in the same row of the ciphertext matrix can not be corrected in general even there. This fact was overlooked in \cite[11.5]{StHar}, and can not be remedied without transmitting additional information. We suggest a particularly natural additional check number, the 
{\it column ratio}, for this purpose. 

The paper is organized as follows. In Section \ref{Gold} we review the basics of golden cryptography, in Section \ref{UniCod} we show how its attractive features can be replicated using arbitrary unimodular matrices. This requires some mathematical excursion into the existence of limits of ratios, which we provide in Section \ref{UniRatio}. In Section \ref{ErrCor} we show how most of the golden error correction carries over to the unimodular cryptography, and in Section \ref{RowErr} we discuss the issues with correcting double errors in a row. Finally, in Section \ref{ColRat} we introduce the column ratio, and demonstrate how it resolves those issues.

\section{Golden Cryptography}\label{Gold}

Consider the matrix $Q=\begin{bmatrix}1 & 1 \\1 & 0 \\\end{bmatrix}$. When this matrix is taken to the $n$-th power the entries are the consecutive Fibonacci numbers: 
\begin{equation}\label{Qn}
Q^n= \begin{bmatrix} F_{n+1} & F_n\\F_n & F_{n-1}\\\end{bmatrix},
\end{equation}
and it follows from their properties that $\det Q^n=(-1)^n$. These matrices Stakhov calls the ``golden matrices", because the ratios of the entries converge to the golden ratio $\tau$ as $n\rightarrow\infty$. The idea of golden cryptography can now be explained as follows. Split the plaintext into a sequence of blocks arranged into groups of four, and use some permutation to place them as the entries of a $2\times2$ plaintext matrix $P=\begin{bmatrix}p_{11}&p_{12}\\p_{21}&p_{22}\end{bmatrix}$. The ciphertext matrix $C$ is then obtained as $C=P\,Q^n$, the chosen permutation and the number $n$ serve as the secret key. To decrypt, one only needs to apply the inverse golden  matrix $Q^{-n}=(-1)^n\begin{bmatrix}F_{n-1}&-F_n\\-F_n&F_{n+1}\end{bmatrix}$, to wit $P=C\,Q^{-n}$. Note that $\det C=(-1)^n\det P$, so $\det P$ can serve as a check number. 
\begin{example}\label{goldexamplex}
Suppose we want to encrypt MATH. Representing each letter by its number in the English alphabet (starting with $0$) we get the plaintext string $12\ \,0\ \,19\ \,7$. It can be arranged into a plaintext matrix 
$P = \begin{bmatrix}12 & 0 \\ 19 & 7\end{bmatrix}$. The check number to be sent to the receiver is $\det P=84$.  For this example we chose $n=10$. Our coding golden matrix becomes \\*
$$
Q^{10} = \begin{bmatrix}F_{11}&F_{10}\\F_{10}&F_{9}\end{bmatrix} = \begin{bmatrix}89&55\\55&34\end{bmatrix}.
$$ 
Multiplying $P$ by it gives the ciphertext matrix:
$$
C =P\,Q^{10}= \begin{bmatrix}12 & 0 \\ 19 & 7\end{bmatrix}\begin{bmatrix}89&55\\55&34\end{bmatrix} 
=\begin{bmatrix}1068&660\\2076&1283\end{bmatrix}.
$$
Comparing the check number to the determinant of the ciphertext matrix, $\det C= 84=\det P$, the receiver can be sure that the code was transmitted correctly. Knowing the key to be $n=10$ the receiver can decrypt the matrix by inverting $Q^{10}$, and multiplying $C$ by the inverse $C \,Q^{-10}$:
$$ 
C\,Q^{-10} = \begin{bmatrix}1068&660\\2076&1283\end{bmatrix}\begin{bmatrix} 34&-55\\-55&89 \end{bmatrix} 
= \begin{bmatrix} 12&0\\19&7\end{bmatrix}=P.
$$ 
\end{example}

As a method of encryption, the golden cryptography is vulnerable to the chosen plaintext attacks \cite{ReySan}. Namely, if 
$P = \begin{bmatrix}1 & 0 \\ 0 & 0\end{bmatrix}$ then $C=P\,Q^n=\begin{bmatrix}F_{n+1} & F_n\\0 & 0\end{bmatrix}$, 
so $F_n$, and then the secret key $n$, can be recovered from $C$, e.g. by using the Binet formula. This will remain the case even if we allow fractional values for $n$, as in the generalized golden cryptography.

Nonetheless, the golden encryption/decryption is fast and has nice error correction properties \cite[11.5]{StHar}. Aside from the determinant check, there are also built-in checking relations in $C$ that do not require transmission of any additional information. They exploit the recurrence properties of Fibonacci numbers instead. Since the entries of $P$ are positive integers for $n$ even from $P=C \,Q^{-n}$ we obtain the following inequalities:
\begin{gather}\label{CheckIneq}
c_{11}F_{n-1} - c_{12}F_{n} \geq 0\notag\\
-c_{11}F_{n} + c_{12}F_{n+1} \geq 0\notag\\
c_{21}F_{n-1} - c_{22}F_{n} \geq 0\\
-c_{21}F_{n} + c_{22}F_{n+1} \geq 0\,,\notag
\end{gather}
which yield
$$ 
\frac{F_{n}}{F_{n-1}} \leq \frac{c_{11}}{c_{12}}\,, \frac{c_{21}}{c_{22}}\leq \frac{F_{n+1}}{F_{n}}\,.
$$
The properties of Fibonacci numbers imply that both ratios converge to the golden ratio $\tau=\frac{1+\sqrt{5}}{2}$, the case of odd $n$ is analogous. So for large enough $n$ of any parity we have approximate relations $c_{11}\approx\tau c_{12}$, $c_{21}\approx\tau c_{22}$. Should one or more of the ciphertext entries be miscommunicated these relations can be used to recover the correct values, see Section \ref{ErrCor}.

\section{Unimodular and coding matrices}\label{UniCod}

Having such a narrow choice of coding options is undesirable in a number of applications (particularly if one wishes to ensure security of encryption). Some generalizations were offered by Stakhov himself \cite{StRoz}, see also \cite{Falc}, \cite[11.3]{StHar}, the so-called generalized $k$-golden cryptography, and by Nalli \cite{Nal}. The golden matrices \eqref{Qn} have rather special, even ``unique", properties, and these generalizations attempt to preserve as many of them as possible. However, not all of them are relevant to cryptography. Let us summarize the properties of $Q^n$ that make them useful for error correction:

\begin{enumerate}

\item The entries are consecutive elements of a sequence $F_n$ satisfying a second order recurrence relation.

\item Determinant is $\pm1$.

\item There is a limit of $\frac{F_{n+1}}{F_n}$ when $n\to\infty$.

\end{enumerate}

A closer look at \eqref{CheckIneq} shows that we do not need to choose entries from a single sequence, a matrix of the form 
$M_n=\begin{bmatrix} A_{n+1} & A_n\\B_{n+1} & B_n\end{bmatrix}$ would work as long as $A_n$, $B_n$ satisfy the same recurrence relation, and $\frac{A_{n+1}}{A_n}$, $\frac{B_{n+1}}{B_n}$ have the same limit. If we set $M_n=U^n$ for some matrix 
$U$ with $\det=\pm1$, i.e. a {\it unimodular matrix}, then the determinant condition will be satisfied as well. While these conditions seem loose they are in fact quite restrictive, as the next theorem shows.
\begin{theorem}\label{UniPower}
Let $U$ be a matrix such that $U^n=\begin{bmatrix} A_{n+1} & A_n\\B_{n+1} & B_n\end{bmatrix}$ for some sequences 
$A_n$, $B_n$. Then either $U$ is degenerate or $U=\begin{bmatrix} \alpha & 1\\\gamma & 0\end{bmatrix}$. In particular, if $U$ is unimodular then $U=\begin{bmatrix} \alpha & 1\\\pm1 & 0\end{bmatrix}$.
\end{theorem}
\begin{proof}
Let $U=\begin{bmatrix} \alpha & \beta\\\gamma & \delta\end{bmatrix}$. Since $U^2=U\cdot U$ we have $\begin{bmatrix}\alpha\\\gamma\end{bmatrix}=U\begin{bmatrix}\beta\\\delta\end{bmatrix}$, or explicitly
$$
\begin{cases}
\alpha=\alpha\beta+\beta\delta\\
\gamma=\gamma\beta+\delta^2
\end{cases}
$$
If $\beta\neq1$ we can solve for $\alpha$, $\gamma$ as follows
$$
\begin{cases}
\alpha=\frac{\beta\delta}{1-\beta}\\
\gamma=\frac{\delta^2}{1-\beta}\,.
\end{cases}
$$
But then $\det U=\frac{\beta\delta}{1-\beta}\,\delta-\frac{\delta^2}{1-\beta}\,\beta=0$, i.e. $U$ is degenerate. If $\beta=1$ then the system reduces to $\alpha=\alpha+\delta$ and $\gamma=\gamma+\delta^2$, i.e. it is satisfied by $\delta=0$ with no conditions on $\alpha$ and $\gamma$. The last claim follows since $\det U=\gamma$.
\end{proof}
With only $\alpha$ as a free parameter we are not that far from the original golden matrix $Q$ with $\alpha=1$. 
Indeed, $Q_k^{\,n}:=\begin{bmatrix}k & 1\\1 & 0\end{bmatrix}^n$ are the so-called $k$-golden matrices of \cite{StRoz,StHar}, whose elements obey a simple generalization of the Fibonacci recurrence relation $F_{n+1}^{(k)}=kF_{n}^{(k)}+F_{n-1}^{(k)}$. As pointed out in \cite{Tah}, encryption by the $k$-golden matrices is vulnerable to the same kind of chosen plaintext attack as encryption by the original golden matrices, it does not even help to allow fractional values of $n$, as in the generalized 
$k$-golden cryptography.

But the peculiarity of Theorem \ref{UniPower} is that severe restrictions on $U$ follow just from the initial multiplication 
$U^2=U\cdot U$. Once this is guaranteed we get $M_n=U^n$ for all $n$ as a bonus. This suggests that instead of taking $M_n$ to be a bare unimodular power, we should modify it by an additional initial factor, i.e. set $M_n=U^nM_0$.
\begin{theorem}\label{Mn}
Let $M_n=U^nM_0$, where $M_0=\begin{bmatrix} A_{1} & A_0\\B_{1} & B_0\end{bmatrix}$ with $A_0$, $B_0$ arbitrarily chosen, and $\begin{bmatrix}A_{1}\\B_{1}\end{bmatrix}=U\begin{bmatrix}A_0\\B_0\end{bmatrix}$. Then $M_n=\begin{bmatrix} A_{n+1} & A_n\\B_{n+1} & B_n\end{bmatrix}$ for all $n\geq0$, and 
\begin{gather}\label{RecurAB}
A_{n+1}=(\tr U)\,A_n-(\det U)\,A_{n-1}\\
B_{n+1}=(\tr U)\,B_n-(\det U)\,B_{n-1}\,.\notag
\end{gather}
\end{theorem}
\begin{proof} Given $A_0$, $B_0$ define $A_n$, $B_n$ by the recurrence relations:
\begin{equation}\label{Recur1AB}
A_{n+1}=\alpha\,A_n+\beta\,B_{n}\text{ and } B_{n+1}=\gamma\,A_n+\delta\,B_{n}\,. 
\end{equation}
Note that this is consistent with the choice of values for $A_1$, $B_1$, and can be rewritten as $\begin{bmatrix}A_{n+1}\\B_{n+1}\end{bmatrix}=U\begin{bmatrix}A_n\\B_n\end{bmatrix}$. Since the latter implies 
$\begin{bmatrix}A_{n}\\B_{n}\end{bmatrix}=U\begin{bmatrix}A_{n-1}\\B_{n-1}\end{bmatrix}$ we have 
$$
R_n:=\begin{bmatrix}A_{n+1}&A_{n}\\B_{n+1}&B_{n}\end{bmatrix}=U\begin{bmatrix}A_n&A_{n-1}\\B_n&B_{n-1}\end{bmatrix}
=UR_{n-1}\,.
$$
But by the definitions, $R_0=M_0$, and $M_n=UM_{n-1}$, therefore $R_n=M_n$ for all $n\geq0$.

From the recurrence relations \eqref{Recur1AB} we have:
\begin{multline}
A_{n+1}=\alpha\,A_n+\beta\,B_{n}=\alpha\,A_n+\beta\,(\gamma\,A_{n-1}+\delta\,B_{n-1})\\
=\alpha\,A_n+\beta\,\gamma\,A_{n-1}+\beta\,\delta\,\frac{A_n-\alpha\,A_{n-1}}{\beta}
=(\alpha+\delta)\,A_n-(\alpha\delta-\beta\,\gamma)\,A_{n-1}\\
=(\tr U)\,A_n-(\det U)\,A_{n-1}\,.\notag
\end{multline}
The case of $B_n$ is analogous.
\end{proof}

Let us denote $\mu:=\det M_0$, then by definition $\det M_n=\mu\,(\det U)^n$, and by Theorem \ref{Mn}:
\begin{equation}\label{mu}
\mu=\det M_0=A_1B_0-A_0B_1=(\alpha-\delta)A_0B_0+\beta B_0^2-\gamma A_0^2\,.
\end{equation}
We could impose an additional constraint to have $\mu=\pm1$, but this is not really necessary for the error correction purposes. If $\det P$ is sent as the check number along with $C$ what matters is that we can independently recover $\det P$ from 
$\det C$, not that it necessarily be equal to it up to sign. Since the encryption scheme is $C=P\,M_n$ we have 
$\det C=(\pm1)^n\mu\det P$, and the recovery is possible for any $\mu$. 

At this point one may be tempted to generalize further by taking 
$M_n=\begin{bmatrix} A_{n+1} & A_n\\B_{n+1} & B_n\end{bmatrix}$ with independently chosen $A_0$, $B_0$, $A_1$ and $B_1$, and 
$A_n$, $B_n$ computed according to \eqref{RecurAB} with $t=\tr U$ and $d=\det U$ not necessarily related to any matrix $U$:
\begin{gather}\label{Recurtd}
A_{n+1}=t\,A_n-d\,A_{n-1}\\
B_{n+1}=t\,B_n-d\,B_{n-1}\,.\notag
\end{gather}
We may not have $M_n=U^nM_0$ anymore, but \eqref{Recurtd} can be rewritten in the matrix form as
$$
\begin{bmatrix}A_{n+1}&A_{n}\\B_{n+1}&B_{n}\end{bmatrix}=\begin{bmatrix}A_n&A_{n-1}\\B_n&B_{n-1}\end{bmatrix}
\begin{bmatrix}t&1\\-d&0\end{bmatrix}\,.
$$
Set $S:=\begin{bmatrix}t&1\\-d&0\end{bmatrix}$, and note that $M_1=M_0\,S$ is automatically satisfied for {\it any} choice of 
$A_0$, $B_0$, $A_1$ and $B_1$, so by induction $M_n=M_0S^n$. 

This is not, however, a real generalization, but rather an alternative representation of $M_n$. Indeed, with the matrix $U$ we could choose six parameters, the four entries of $U$ and $A_0$, $B_0$, while $A_1$ and $B_1$ were then determined by them. With the matrix $S$ we get to choose all four entries of $M_0$, but in turn it only has two parameters in it, $t$ and $d$, which still total six. In other words, we simply traded the diversity of $U$ for the freedom in the choice of $M_0$.

The main reason the golden and the $k$-golden cryptographies are vulnerable to chosen plaintext attacks \cite{ReySan,Tah} is that the system has too few free parameters, $n$ and $n,k$ respectively. As shown in \cite{Tah}, the more parameters are introduced the harder it is to recover them from equations obtained by choosing special plaintexts. The authors of \cite{Tah} inject some extra parameters by introducing an external hash function, and ask for a more intrinsic way of doing so (other suggestions were made in \cite{Moh,Sud}). We believe that the $M_n$ coding matrices provide just that. For cryptographic intents and purposes, $M_n$ share the first two properties of the golden matrices listed at the beginning of this section, with the added benefit of parameter choice freedom. In the next section we will show that they share the third property as well.

\section{Unimodular ratio}\label{UniRatio}

By Theorem \ref{Mn}, $A_n$, $B_n$ satisfy a second order recurrence relation $A_{n+1}=tA_n-dA_{n-1}$, where $t=\tr U$ and 
$d=\det U$. Since we are interested in convergence of the ratios let us rewrite it in terms of them: 
$$
\frac{A_{n+1}}{A_n}=t-\,\frac{d}{\frac{A_n}{A_{n-1}}}\,.
$$
Assuming that $\frac{A_{n+1}}{A_n}\xrightarrow[n\to\infty]{}\varphi$ we have that $\varphi$ satisfies $x=t-\,\frac{d}{x}$, or 
\begin{equation}\label{UniRatEq}
x^2-tx+d=0\,.
\end{equation}
The obvious problem is that this is a quadratic equation which generically has two different roots, and the ratios can not converge to both at once. Of course, this was already the situation with the Fibonacci numbers, but there only one of the roots was positive, and since the Fibonacci numbers are positive this is the one to which their ratios converged. But for $d=\det U>0$ equation \eqref{UniRatEq} {\it may} have two positive roots. 

Let $f(x):=t-\,\frac{d}{x}$, and $a_n:=\frac{A_{n+1}}{A_n}$, then the ratio recurrence becomes $a_{n+1}=t+\,\frac{d}{a_n}$, and 
the roots of \eqref{UniRatEq},
\begin{equation}\label{Pfi+-}
\varphi_\pm:=\frac{t\pm\sqrt{t^2-4d}}{2},
\end{equation}
are the fixed points of $f$. Clearly, if $a_n$ converge at all the limit would have to be a fixed point, i.e. $\varphi_+$ or 
$\varphi_-$. If $d>0$ they are real only for $|t|\geq2\sqrt{d}$, and then both positive. Otherwise $f$ has no real fixed points, and the ratios diverge. We will write $\nearrow$ and $\searrow$ to denote monotone convergence, increasing and decreasing, respectively.
\begin{theorem}\label{Convd+} Let $a_{n+1}=f(a_n)=t-\,\frac{d}{a_n}$ with $d>0$ and $t\geq2d$.\\
{\rm(i)\,} If $a_0\geq\varphi_+$ then $a_n\searrow\varphi_+$.\\
{\rm(ii)} If $\varphi_-<a_0<\varphi_+$ then $a_n\nearrow\varphi_+$.

In particular, for any $a_0>\varphi_-$ we have $a_n\to\varphi_+$.
\end{theorem}
\begin{proof}  
Note that for $x>0$ the function $f(x)$ is monotone increasing. Therefore, 
for $a_0\geq\varphi_+$ we have $a_{1}=f(a_0)\geq f(\varphi_+)=\varphi_+$, and 
$$
a_{1}-a_0=t-\,\frac{d}{a_0}-a_0=-\frac1{a_0}(a_0^2-ta_0+d)=-\frac1{a_0}(a_0-\varphi_-)(a_0-\varphi_+)\leq0\,.
$$
Taking $a_1$ as the new $a_0$, and so on, we get by induction $a_0\geq a_1\geq a_2\geq\dots\geq\varphi_+$. Thus, $a_n$ is monotone decreasing and bounded below by $\varphi_+$, therefore it converges. The limit is a fixed point of $f$, and therefore $\varphi_+$ since $\varphi_-<\varphi_+$. This concludes the proof of {\rm(i)}. The proof of {\rm(ii)} is analogous.
\end{proof}
\begin{theorem}\label{Convd-} Let $a_{n+1}=f(a_n)=t-\,\frac{d}{a_n}$ with $d<0$ and $t>0$.\\
{\rm(i)\,} If $a_0\geq\varphi_+$ then $a_{2k}\searrow\varphi_+$, $a_{2k+1}\nearrow\varphi_+$.\\
{\rm(ii)} If $0<a_0\leq\varphi_+$ then $a_{2k}\nearrow\varphi_+$, $a_{2k+1}\searrow\varphi_+$.

In particular, for any $a_0>0$ we have $a_n\to\varphi_+$.
\end{theorem}
\begin{proof}  
This time for $x>0$ the function $f(x)$ is monotone decreasing, so if $a_0\geq\varphi_+$ then $a_{1}=f(a_0)\leq f(\varphi_+)=\varphi_+$, and 
\begin{multline*}
a_{2}-a_0=t-\,\frac{d}{t-\,\frac{d}{a_0}}-a_0=-\frac{t}{1+ta_0}(a_0^2-ta_0+d)\\
=-\frac{t}{1+ta_0}(a_0-\varphi_-)(a_0-\varphi_+)\leq0\,.
\end{multline*}
On the other hand, if $0<a_0\leq\varphi_+$ by the same reasoning $a_{1}\geq\varphi_+$, and $a_{2}\geq a_0$.

Consider $a_0\geq\varphi_+$ again. Taking $a_1$ as the new $a_0$ we get $a_2\geq\varphi_+$, and so on. By induction, $a_0\geq a_2\geq a_4\geq\dots\geq\varphi_+$. Thus, $a_{2k}$ is monotone decreasing and bounded below by $\varphi_+$, therefore it converges. The limit is a fixed point of $f$, and therefore $\varphi_+$ since $\varphi_-<\varphi_+$. The case of $a_{2k+1}$, and the proof of {\rm(ii)} are analogous.
\end{proof}
As a matter of fact, one can show that in both cases $a_n\to\varphi_+$ for any $a_0\neq0,\varphi_-$, essentially because 
$\varphi_-$ is a repulsive fixed point. But the convergence may not be monotone since for $a_0$ close to $\varphi_-$ the sequence may spend a long time in its vicinity before getting to approach $\varphi_+$. In our conditions (with $t>2\sqrt{d}$ for $d>0$) one can show that the convergence is exponential, i.e. $a_n$ approaches $\varphi_+$ faster than some geometric sequence approaches $0$. From now on we drop $+$ from the notation and call
\begin{equation}\label{Pfi}
\varphi:=\frac{t+\sqrt{t^2-4d}}{2},
\end{equation}
the {\it unimodular ratio}.
\begin{corollary}\label{LimRat} Suppose the entries of $U$ are non-negative, and $A_0,B_0\geq1$. Then for $d=1$ we have 
$\frac{A_{n+1}}{A_n},\frac{B_{n+1}}{B_n}\xrightarrow[n\to\infty]{}\varphi$, and the convergence is exponential. The same holds for $d=-1$ if additionally 
$\tr U>2$ and $\alpha,\delta\geq1$.
\end{corollary}
\begin{proof}  
The conditions ensure that Theorems \ref{Convd+}, \ref{Convd-} apply, in particular that $\frac{A_{1}}{A_0},\frac{B_{1}}{B_0}>1>\varphi_-$.
\end{proof}

\section{Error correction}\label{ErrCor}

Encryption and decryption in the unimodular cryptography is completely analogous to that in the golden cryptography. Namely, 
the plaintext is divided into blocks of four arranged into a matrix, the encryption is given by $C=P\,M_n$, and
the decryption by $P=C\,M_n^{-1}$. The same goes for the error correction described in \cite[11.5]{StHar}. 
This is based on the determinant formula $\det M_n=\mu d^n$, and the {\it unimodular ratio checking relations} analogous to the ones of the golden cryptography: $c_{11}\approx\varphi c_{12}$, $c_{21}\approx\varphi c_{22}$. The proof is analogous to the golden case (given convergence of the ratios proved above), see \eqref{CheckIneq}, and we omit it. 

Let us briefly outline Stakhov's error correction procedures. There are four possible locations for a single error, represented by variables in the following matrices:
$$ 
\begin{bmatrix}x&c_{12}\\c_{21}&c_{22}\end{bmatrix}, \hspace{24pt} \begin{bmatrix}c_{11}&y\\c_{21}&c_{22}\end{bmatrix}, \hspace{24pt} \begin{bmatrix}c_{11}&c_{12}\\z&c_{22}\end{bmatrix},\hspace{12pt} \textrm{and} \hspace{12pt} \begin{bmatrix}c_{11}&c_{12}\\c_{21}&v\end{bmatrix}.
$$
The most effective method of single error correction involves using the unimodular ratio checking relations first. The receiver can determine in which row the single error is located by determining which row does not satisfy the checking relation. Next, the two possible incorrect elements can be estimated by assuming the other one to be correct. Finally, the correct solution will satisfy the expected determinant equation:
\begin{align*}
xc_{22}-c_{12}c_{21}&=\mu d^n\det P\\
c_{11}c_{22}-yc_{21}&=\mu d^n\det P\\
c_{11}c_{22}-c_{12}z&=\mu d^n\det P\\
c_{11}v-c_{12}c_{21}&=\mu d^n\det P.
\end{align*}
The following example demonstrates how a single error is detected and corrected.
\vspace{0.5cm}

\begin{example}\label{SinglCor}
For encryption we will use the following matrices:
$$
U := \begin{bmatrix}2 & 1 \\ 1 & 1\end{bmatrix}\text{ and } M_0=\begin{bmatrix}1 & 0 \\ 1 & 1\end{bmatrix}\,.  
$$
The unimodular matrix $U$ is known as the Arnold's cat matrix \cite{Mish} (taken $\mod{1}$ it generates a chaotic map discovered by Arnold in 1960-s, who illustrated its action on an image of a cat). We have $t=3$, $d=1$, and 
$$
\varphi=\frac{3+\sqrt{5}}{2}=1+\tau\approx2.618...\,.
$$
Suppose the following ciphertext matrix, encrypted with $M_n=U^nM_0$ with $n=4$, was received,
$$
C=\begin{bmatrix} 770 & 494 \\ 1846 & 705 \end{bmatrix},
$$
along with the determinant check number, $126$. We expect the determinant of the ciphertext matrix to be $126$. However $C$ has the determinant of $-369074,$ thus there must be an error in it. In order to determine where the error is located, we compute the row ratios:
\begin{align*}
\frac{770}{494} &\approx 1.559\,, \\
\frac{1846}{705} &\approx 2.618.
\end{align*}
We see that the top row of the ciphertext matrix does not satisfy the checking relation, but the bottom row does. We assume that there is a single error in the top row and the bottom row is correct. First, we replace the element $c_{11}=770$ by the variable $x$ and find an estimate for $x$ with the unimodular checking relation:
$$
x \approx 2.618\cdot 494 \approx 1293.
$$
If we let $c_{11}=x$ and check the determinant of the new matrix, we find that the determinant equals $-359\neq126$, so
$c_{11}\neq1293$. Next, we replace the element $c_{12}=494$ by the variable $y$ and estimate $y$ with the unimodular checking relation:
$$
y \approx \frac{770}{2.618} \approx 294.
$$
Now the determinant equals 126, which means $c_{12}=294$. The corrected matrix is
$$
C=\begin{bmatrix} 770 & 294 \\ 1846 & 705 \end{bmatrix}.
$$
\end{example}\vspace{0.5cm}

For double error correction it is crucial that the check number $\det P$ be transmitted correctly. The determinant of the ciphertext matrix is equated to $\mu d^n\det P$, and the two elements assumed incorrect are solved for by applying the unimodular ratio checking relations. Double errors fall into three categories: diagonal, column and row errors.
There are two types of the diagonal double errors in the ciphertext matrix, the diagonal error and the anti-diagonal error, respectively
$$
\begin{bmatrix}x&c_{12}\\c_{21}&v\end{bmatrix} \hspace{24pt}\textrm{and}\hspace{24pt} 
\begin{bmatrix}c_{11}&y\\z&c_{22}\end{bmatrix}.
$$
Both cases yield a factoring problem:
\begin{align*}
xv &= c_{12}c_{21} + \mu d^n\det P \text{\ \ \ \ and}\\
yz &= c_{11}c_{22} - \mu d^n\det P.
\end{align*}
The right hand side of each equation is known, thus the correct solution is one of its factor pairs. If $n$ was chosen to be sufficiently large this factoring problem can be rather challenging, as the number to be factored should be quite large. But with high probability only the correct solution satisfies the unimodular ratio checking relations. This simplifies the factoring problem considerably. For the errors along the diagonal the solutions are approximated by 
$$
x \approx \varphi c_{12} \textrm{ and } v \approx \frac{c_{21}}{\varphi}\,;
$$
and for the errors along the anti-diagonal by
$$ 
y \approx \frac{c_{11}}{\varphi} \textrm{ and } z \approx \varphi c_{22}.
$$
With these estimates the receiver only needs to try dividing the product by numbers close to them, as opposed to finding all factor pairs. If there are no factor pairs around the estimated values, a different double error should be assumed.

There are two types of double column errors in the ciphertext matrix, 
$$
\begin{bmatrix}x&c_{12}\\z&c_{22}\end{bmatrix} \hspace{24pt}\textrm{and}\hspace{24pt}  
\begin{bmatrix}c_{11}&y\\c_{21}&v\end{bmatrix},
$$
which lead to solving the following linear Diophantine equations, respectively
\begin{align*}
xc_{22} - c_{12}z &= \mu d^n\det P \text{\ \ \ \ and}\\
c_{11}v - yc_{21} &= \mu d^n\det P.
\end{align*}
There are infinitely many solution pairs that would satisfy them. The desired solution will also satisfy the unimodular ratio checking relations. The receiver can estimate the solutions as
\begin{align*}
x \approx \varphi c_{12} \textrm{ and } z \approx \varphi c_{22}\,;\\
y \approx \frac{c_{11}}{\varphi} \textrm{ and } v\approx \frac{c_{21}}{\varphi}.
\end{align*}
Again, one can search for a solution pair to the Diophantine equations near the estimates. If there are no Diophantine solutions near the estimates, a different type of error should be assumed.

\section{Row Errors}\label{RowErr}

It may seem that the correction of row errors should be analogous to the correction of column errors, but the symmetry is broken by the fact that the unimodular ratio checking relations relate entries in a row, not in a column. The issue arises already in the golden cryptography, but it is overlooked in \cite[11.5]{StHar}. As we will see, without additional information double row errors can not be corrected. There are two cases of double row errors in the ciphertext matrix, 
$$
\begin{bmatrix}x&y\\c_{21}&c_{22}\end{bmatrix}\hspace{24pt}\textrm{and}\hspace{24pt}
\begin{bmatrix}c_{11}&c_{12}\\z&v\end{bmatrix},
$$ 
which lead to the following linear Diophantine equations, respectively
\begin{align}\label{RowEq}
xc_{22} - yc_{21} &= \mu d^n\det P  \text{\ \ \ \ and}\notag\\
c_{11}v - c_{12}z &= \mu d^n\det P.
\end{align}
With the diagonal and column double errors we had at least one correct element in each row, which allowed us to identify the correct solution more or less uniquely. But if both elements in a row are faulty, the checking relations generally do not narrow down the available choices sufficiently. Indeed, in the first case, say, we have 
$$
\left|\frac{x}{y}-\frac{c_{21}}{c_{22}}\right|=\mu d^n\,\frac{\det P}{yc_{22}}.
$$ 
Suppose $|\mu|=|d|=1$, as for the $Q$ matrix or the Arnold's cat matrix from Example \ref{SinglCor}, and $n$ is large. Then  the entries of $C$ will be much larger than the entries of $P$, and therefore the ratio on the right will be small for {\it all} solutions to the Diophantine equation \eqref{RowEq}. But that means that the ratios of all solutions to the equation will be close to the unimodular ratio $\varphi$, and the corresponding checking relation is of no help in correcting the transmitted data.

\vspace{0.5cm}\begin{example}\label{qrowratioex}
Suppose the matrix $P=\begin{bmatrix}8 & 24\\19&2\end{bmatrix}$ is encrypted with the $Q$ matrix and the key $n=6$.  The deterimant of $P$ is $-440$.  If the top row of the ciphertext matrix contains errors,
$$
C=\begin{bmatrix}x&y\\263&162\end{bmatrix}\,,
$$
the Diophantine equation to solve is $162x-263y=-440$. The solutions are $x=263k+33$ and $y=162k+22$, with $k\in\mathbb{Z}$.  The following table displays low positive $k$ solutions to the Diophantine equations, and the ratios $\frac{x}{y}$ for each solution.
\begin{center}
\begin{tabular}{|c||c|c||c|}
\hline $k$ & $x$ & $y$ & $\frac{x}{y}$ \\ [0.5ex]
\hline 0 & 33 & 22 & 1.50000\\
\hline 1 & 296 & 184 & 1.60870 \\
\hline 2 & 559 & 349 & 1.61561 \\
\hline 3 & 822 & 508 & 1.61811 \\
\hline 4 & 1085 & 670 & 1.61940 \\
\hline 5 & 1348 & 832 & 1.62019\\ \hline
\end{tabular}
\end{center}
As expected, the ratios are all quite close to the golden ratio $1.618...$\,. The correct solution, based on the original encryption, has $k=1$. However, the ratio closest to the golden ratio occurs when $k=3$. This demonstrates that the row ratio checking relations can not correct the double row errors even in the original golden cryptography.
\end{example}\vspace{0.5cm}

Depending on the situation, the receiver may be able to rule out the unlikely solutions and pick the most likely one. Because we assumed that all plaintext elements were greater than or equal to zero, the negative solution pairs are ruled out.  Additionally, the sizes of the ciphertext elements depend on plaintext elements and the entries of the coding matrices. If the range for the plaintext elements is known, the range for the ciphertext elements can be found. However, this method of correction relies on all four plaintext entries being within a very small range. For example, if the plaintext entries represent letters of the alphabet, coded from $0$ to $25$, the receiver can determine the expected range of ciphertext entries and narrow down the acceptable Diophantine solutions. Still, too many solutions may be left, as the following example demonstrates.

\vspace{0.5cm}\begin{example}
Suppose the matrix $P=\begin{bmatrix}19&7\\2&10\end{bmatrix}$ is encrypted using the Arnold's cat matrix and the initial matrix $M_0=\begin{bmatrix}1&0\\1&1\end{bmatrix}$ with a key of $n=4$. The check number sent is 176. The receiver determined that the top row must contain a double error:
$$
C = \begin{bmatrix}x&y\\450&172\end{bmatrix}.
$$
The resulting Diophantine equation, $172x-450y=176$ has solutions $x=225k+158$ and $y=86k+60$, with $k\in\mathbb{Z}$.  Because the plaintext elements are known to range from $0$ to $25$ we have,
	\begin{align*}
	0 &\leq x \leq 25(A_{n+1}+B_{n+1}) \text{\ \ \ \ and}\\
	0 &\leq y \leq 25(A_n+B_n).
	\end{align*}
Thus the receiver can determine the following bounds on the possible solutions,
	\begin{align*}
	0 &\leq x \leq 2225  \text{\ \ \ \ and}\\
	0 &\leq y \leq 850.
	\end{align*}
Nonetheless, there are still ten values of $k$ which give Diophantine solutions $x$ and $y$ within this range. Moreover, narrowing the range of plaintext entries to from $0$ to $25$ is likely to compromise security.
\end{example}

\section{Column ratio}\label{ColRat}

Although we do not have ratios of column entries converge to the unimodular ratio, it turns out that the ratios in both columns converge to the same value, we call this common value the {\it column ratio}.
\begin{corollary}\label{LimColRat}
In conditions of Corollary \ref{LimRat} we have $\frac{c_{21}}{c_{11}} \approx \frac{c_{22}}{c_{12}}$ for large $n$.
\end{corollary}
\begin{proof}
By definition of encryption, we have for the ratios of the column entries:
\begin{align*}
\frac{c_{21}}{c_{11}} &= \frac{ p_{21}A_{n+1}+p_{22}B_{n+1}}{p_{11}A_{n+1}+p_{12}B_{n+1}} = \frac{\frac{A_{n+1}}{B_{n+1}}p_{21} + p_{22}}{\frac{A_{n+1}}{B_{n+1}}p_{11}+p_{22}}\,;\\
\frac{c_{22}}{c_{12}} &= \frac{p_{21}A_{n}+p_{22}B_{n}}{p_{11}A_{n}+p_{12}B_{n} } = \frac{\frac{A_n}{B_n}p_{21}+p_{22}}{\frac{A_n}{B_n}p_{11}+p_{12}}\,.
\end{align*}
Since in conditions of Corollary \ref{LimRat} the ratios $\frac{A_{n+1}}{A_n}$ and $\frac{B_{n+1}}{B_n}$ converge to the same limit we also have $\frac{A_{n+1}}{B_{n+1}} \approx \frac{A_n}{B_n}$ for large $n$. But this implies 
$\frac{c_{21}}{c_{11}} \approx \frac{c_{22}}{c_{12}}.$
\end{proof}
Unlike the unimodular ratio, which only depended on $A_n$ and $B_n$, the column ratio also depends on the plaintext matrix $P$.
The next example demonstrates how much the column ratio can vary.

\vspace{0.5cm}\begin{example}
Of the two plaintext matrices below, $P_1$ has elements in a relatively small range, whereas $P_2$ has the same bottom row elements, but larger top row elements.  Observe how the resulting ciphertext matrices reflect those differences.
	\begin{align*}
		P_1\,M_n = \begin{bmatrix}7 & 8 \\ 3 & 5\end{bmatrix} &\times \begin{bmatrix} 21&8\\13&5 \end{bmatrix} = \begin{bmatrix} 251&96\\128&49 \end{bmatrix} = C_1\,, \\
		P_2\,M_n = \begin{bmatrix}56 & 45 \\ 3 & 5\end{bmatrix} &\times \begin{bmatrix} 21&8\\13&5 \end{bmatrix} = \begin{bmatrix} 1761&673\\128&49 \end{bmatrix} = C_2\,.
	\end{align*}
In $C_1$ the column ratio is approximately 1.96, in $C_2$ it is approximately 13.8. Thus, if only one row is known in the ciphertext matrix, there is no way to estimate the magnitude of the other row without knowing the ratios of the column elements.
\end{example}\vspace{0.5cm}

To allow the row double error correction one needs to send an additional check number, e.g. the column ratio. Even modest precision in it suffices to recover the mistransmitted row elements. Of course, there is a trade-off involved as it increases the size and reduces the security of the required transmission. The following examples demonstrate the double row error correction with the column ratio.

\vspace{0.5cm}\begin{example}
Recall the faulty ciphertext matrix from Example \ref{qrowratioex},
$$
C=\begin{bmatrix} x&y \\ 263 & 162\end{bmatrix}.
$$
Suppose that in addition to the Diophantine equation $162x-263y=-440$ the column ratio is known $\frac{c_{21}}{c_{11}} \approx \frac{c_{22}}{c_{12}} \approx 0.9$. Then the solutions can be estimated as $x \approx \frac{263}{0.9} \approx 292$ and $y \approx \frac{162}{0.9} \approx 180$. As seen from the table in Example \ref{qrowratioex}, the solution pair closest to the estimates clearly has $k=1$. Thus, the corrected matrix is
$$ 
C = \begin{bmatrix} 296 & 184 \\ 263 & 162 \end{bmatrix}.
$$
\end{example}

Similarly, the column ratio can be applied to a ciphertext matrix encrypted with a unimodular matrix.

\vspace{0.5cm}\begin{example}
Suppose that the plaintext matrix
$$
P = \begin{bmatrix} 14 & 20 \\ 9 & 7 \end{bmatrix}
$$ 
is encrypted with the Arnold's cat matrix, $n=4$, and the initial matrix $M_0=\begin{bmatrix}1&0\\1&1\end{bmatrix}$.  The ciphertext matrix is sent via a noisy channel and received as
$$
\widetilde{C} = \begin{bmatrix} 1325 & 321 \\ 733 & 280 \end{bmatrix}
$$
with check numbers $\det P = -82$ and $\frac{c_{21}}{c_{11}} \approx \frac{c_{22}}{c_{12}} \approx 0.5$. The receiver finds that the top row does not satisfy the unimodular checking relation, and the single error correction methods do not correct the matrix. Finally, it is assumed that the top row contains double errors, and
$$
C = \begin{bmatrix} x & y \\ 733 & 280 \end{bmatrix}.
$$
This gives the Diophantine equation $280x-733y=-82$.  Using the column ratio 0.5, we estimate $x \approx \frac{733}{0.5} \approx 1466$ and $y \approx \frac{280}{0.5} \approx 560$. Some possible solution pairs to the Diophantine equation around the estimates are $(717, 274)$, $(1450,554)$, $(2183,834)$ and $(2916,1114)$. Clearly, the solution nearest the approximation is 
$x=1450$ and $y=554$. Letting $c_{11}=1450$ and $c_{12}=554$ the correct ciphertext matrix is obtained:
$$ 
C =  \begin{bmatrix} 1450 & 544 \\ 733 & 280 \end{bmatrix}.
$$
\end{example}

\section{Conclusions}

We introduced a generalization of the golden cryptography that preserves its error correction benefits while increasing security of encryption by using extra free parameters. In particular, the generalization is not susceptible to the known 
types of chosen plaintext attacks. While all of the golden error correction carries over to the unimodular case, we uncovered that correction of double row errors is problematic already there, and offered a solution based on sending an additional check number, the column ratio.

Although the unimodular cryptography is more secure than the golden cryptography in terms of known attacks, further research is needed to make sure that it can not be compromised in some more elaborate ways. The effect of transmitting the column ratio on security of encryption also needs to be investigated. It is likely that additional layers of encryption, such as those suggested in \cite{Moh,Sud,Tah}, are needed to secure even the unimodular cryptography transmissions.

\end{document}